\documentclass[a4paper]{article}
 
\usepackage{a4wide}
\usepackage[utf8]{inputenc}
\usepackage{amsfonts}
\usepackage{amsmath, amsthm}
\usepackage{enumitem}
\usepackage{graphicx}
\usepackage{microtype}

\theoremstyle{definition}
\newtheorem{theorem}{Theorem}
\newtheorem{lemma}[theorem]{Lemma}
\newtheorem{proposition}[theorem]{Proposition}

\newtheorem{remark}[theorem]{Remark}

\newcommand{\cI}{\mathcal{I}}
\newcommand{\cM}{\mathcal{M}}

\newcommand{\Real}{\mathbb{R}}
\newcommand{\Nat}{\mathbb{N}}

\newcommand{\<}{\langle}
\renewcommand{\>}{\rangle}

\newcommand{\Mapsto}{{\mathop{\mapsto}}}

\newcommand{\phiChunk}{\psi_{001011}}
\newcommand{\phiBreak}{\psi_{0011}}
\newcommand{\phiStart}{\psi_{01}}
\newcommand{\phiEnd}{\psi_{10}}

\newcommand{\PhiChunk}{\Psi_{001011}}

\newcommand{\PhiEnd}{\Psi_{10}}
\newcommand{\Chi}{X}

\newcommand{\Add}{\text{add}}
\newcommand{\Succ}{\text{succ}}
\newcommand{\Times}{\text{times}}

\begin{document}
\title{The Universal Fragment of Presburger Arithmetic with Unary Uninterpreted Predicates is Undecidable}
\author{
	\begin{tabular}{l}
		Matthias Horbach\\
		\small\textit{Max Planck Institute for Informatics, Saarland Informatics Campus, Saarbr\"ucken, Germany}
	\end{tabular}
	\and
	\begin{tabular}{l}
		Marco Voigt\\
		\small\textit{Max Planck Institute for Informatics, Saarland Informatics Campus, Saarbr\"ucken, Germany,}\\
		\small\textit{Saarbr\"ucken Graduate School of Computer Science}
	\end{tabular}
	\and
	\begin{tabular}{l}
		Christoph Weidenbach \\
		\small\textit{Max Planck Institute for Informatics, Saarland Informatics Campus, Saarbr\"ucken, Germany}
	\end{tabular}
}	
\date{}
\maketitle

\begin{abstract}
	The first-order theory of addition over the natural numbers, known as Presburger arithmetic, is decidable in double exponential time.
	Adding an uninterpreted unary predicate to the language leads to an undecidable theory.
	We sharpen the known boundary between decidable and undecidable in that we show that the purely universal fragment of the extended theory is already undecidable.
	Our proof is based on a reduction of the halting problem for two-counter machines to unsatisfiability of sentences in the extended language of Presburger arithmetic that does not use existential quantification.
	On the other hand, we argue that a single $\forall\exists$ quantifier alternation turns the set of satisfiable sentences of the extended language into a $\Sigma^1_1$-complete set.
	
	Some of the mentioned results can be transfered to the realm of linear arithmetic over the ordered real numbers.
	This concerns the undecidability of the purely universal fragment and the $\Sigma^1_1$-hardness for sentences with at least one quantifier alternation.
	
	Finally, we discuss the relevance of our results to verification. 
	In particular, we derive undecidability results for quantified fragments of separation logic, the theory of arrays, and combinations of the theory of equality over uninterpreted functions with restricted forms of integer arithmetic. 
	In certain cases our results even imply the absence of sound and complete deductive calculi.
\end{abstract}

\section{Introduction}\label{section:Introduction}

In 1929 Moj\. zesz Presburger presented a quantifier elimination procedure that decides validity of first-order sentences over the natural numbers with addition \cite{Presburger1929} (see \cite{Stansifer1984} for an English translation and \cite{Enderton2001} for a textbook exposition).
Today, this theory is known as \emph{Presburger} arithmetic. 
In 1974 the computational time complexity of deciding validity of its sentences was shown to be double exponential by Fischer and Rabin~\cite{Fischer1974}.
It has been proved in several ways that the addition of a single uninterpreted unary predicate symbol to the language renders the validity problem undecidable.
In 1957 Putnam \cite{Putnam1957} discussed this theory as one example of an undecidable theory that is somewhat stronger than the decidable theory of natural numbers with the successor function and a uninterpreted unary predicates. 
Lifshits mentioned  in a note \cite{Lifshits1969} (without giving a proof) that the addition of one predicate---of unspecified arity---to Presburger arithmetic leads to undecidability.
In a technical report \cite{Downey1972} from 1972 Downey gave an encoding of two-counter machines and their halting problem in Presburger arithmetic with a single unary predicate symbol.
Moreover, undecidability is also implied by a general result due to Garfunkel and Schmerl \cite{Garfunkel1974} published in 1974.
Seventeen years later Halpern \cite{Halpern1991} strengthened the undecidability result in that he proved $\Pi^1_1$-completeness of this problem.
Only recently, Speranski \cite{Speranski2013a} gave an alternative characterization of the analytical hierarchy that is based on a reduction of $\Pi^1_n$-formulas with multiplication to $\Pi^1_n$-formulas without multiplication. Halpern's $\Pi^1_1$-completeness can be read as a special case of this more general point of view.

Halpern's proof rests on a result by Harel, Pnueli and Stavi (Proposition 5.1 in \cite{Harel1983}), which states that the set of G\"odel numbers of recurring Turing machines is $\Sigma^1_1$-complete.\footnote{Halpern's proof shifts the perspective from the validity problem to the problem of satisfiability. A $\Sigma^1_1$-complete satisfiability problem entails a $\Pi^1_1$-complete validity problem and vice versa, given that the considered languages are closed under negation. For the definition of the analytical hierarchy and the sets $\Pi^1_1$ and $\Sigma^1_1$, see, e.g., Chapter~IV.2 in \cite{Odifreddi1992} or Chapter~16 in \cite{Rogers1987}.} A nondeterministic Turing machine is considered to be \emph{recurring} if, started on an empty input tape, it is able to perform a nonterminating computation in which it infinitely often reaches its initial state (but not necessarily its initial configuration).
The encoding of recurring Turing machines that Halpern employs in his proof results in formulas with two quantifier alternations. More precisely, the used sentences start with a $\forall^* \exists^* \forall^*$-prefix of first-order quantifiers when written in prenex normal form. 
The reduction by Speranski \cite{Speranski2013a} relies on the same pattern of quantifier alternations. 
However, the required quantifier alternation in Halpern's proof can be simplified to $\forall^* \exists^*$, as pointed out by Speranski in \cite{Speranski2013b}.
Formally, Downey's encoding of two-counter machines in \cite{Downey1972} exhibits a $\forall \exists$ alternation as well. However, in this case, suitable modifications lead to an encoding that does not require existential quantification.
A crucial difference between Downey's encoding and ours is that the former concentrates on reachability of configurations, while the latter also considers the temporal order in which configurations are reached. 
One consequence is that our encoding facilitates the formalization  of \emph{recurrence} for nondeterministic two-counter machines. 
This requires some chronological information regarding the configurations that occur in a run that goes beyond reachability.

In the main part of the present paper we restrict the admitted language so that only universal first-order quantifiers may be used. Yet, the resulting validity and satisfiability problems remain undecidable.\footnote{In fact, this result can be obtained from Downey's proof \cite{Downey1972}---even for Horn clauses---after suitable modifications to his encoding of two-counter machines. 
Apparently, Downey was not concerned with minimizing quantifier prefixes.}
To be more precise, we show $\Sigma^0_1$-completeness of the set of unsatisfiable sentences from the universal fragment of Presburger arithmetic extended with a single uninterpreted unary predicate symbol (cf.\ Theorems~\ref{theorem:UndecidabilityOverIntegers} and \ref{theorem:UnsatUnivPresSigmaZeroOneComplete}). As it turns out, this result is still valid when we use the reals as the underlying domain (Theorem~\ref{theorem:UndecidabilityOverReals}).

Our proof proceeds by a reduction of the (negated) halting problem for two-counter machines (cf.\ \cite{Minsky1967}) to the satisfiability problem in the described language.
A run of such a machine started with a certain input can be represented by a (potentially infinite) sequence of configurations $\<\ell, c_1, c_2\>$---triples of natural numbers---, where $\ell$ describes the control state of the machine and $c_1$, $c_2$ are the current values of the machine's counters.
It is not very hard to imagine that such a sequence of configurations can be encoded by a (potentially infinite) string of bits.
On the other hand, we can conceive any interpretation of a unary predicate over the natural numbers as a bit string. 
Given this basic idea, it remains to devise a translation of the program of an arbitrary two-counter machine into a suitable set of sentences from the universal fragment of Presburger arithmetic with an additional uninterpreted unary predicate symbol $P$. Suitable in this case means that any model of the resulting set of formulas interprets $P$ such that it faithfully represents a run of the given machine on the given input. Section~\ref{section:Encoding} is devoted to exactly this purpose. In Section~\ref{section:Preliminaries} we recap the necessary preliminaries.

In Section~\ref{section:SigmaOneOneCompleteness} we relax our language restrictions a bit and show that allowing one quantifier alternation entails a high degree of undecidability.
More precisely, the set of satisfiable $\forall^* \exists^2$-sentences is $\Sigma^1_1$-complete. 
The proof rests on a lemma, due to Alur and Henzinger \cite{Alur1994}, that rephrases Harel et al.'s $\Sigma^1_1$-hardness result for recurring Turing machines in terms of recurring two-counter machines.
In order to apply this lemma, we have to adapt the encoding presented in Section~\ref{section:Encoding} only slightly. All we need to do is to add the possibility of nondeterministic branching of the control flow and to replace the check for the reachability of the \texttt{halt} instruction by a condition that formalizes the recurrence property. 

Moreover, we observe that our undecidability and $\Sigma^1_1$-hardness results for settings over the integer domain can be transfered to corresponding results in the realm of real numbers. We do so at the end of Sections~\ref{section:Encoding} and \ref{section:SigmaOneOneCompleteness}, respectively.

Finally, we discuss the relevance of our findings to the field of verification in Section~\ref{section:Verification}. 
In particular, we derive undecidability results for quantified fragments of separation logic (Section~\ref{section:VerificationSeparationLogic}), the theory of arrays (Section~\ref{section:VerificationDataStructures}), and combinations of the theory of equality over uninterpreted functions with restricted forms of integer arithmetic (Sections~\ref{section:VerificationEUFwithCounterArithmetic} and \ref{section:VerificationEUFwithOffsets}). 
In certain cases our results even imply the absence of sound and complete deductive calculi.

The authors would like to stress that all of the results outlined above are obtained based on refinements of the encoding of two-counter machines presented in Section~\ref{section:FirstEncoding}. 
To the authors' knowledge, a similarly general applicability is not documented for any other encoding of hard problems in the language of Presburger arithmetic augmented with uninterpreted predicate symbols.

\section{Preliminaries}\label{section:Preliminaries}
\subsection{The universal fragment of Presburger arithmetic}\label{section:basicDefinitionsPresburger}

We define the language of \emph{Presburger arithmetic} to comprise all first-order formulas with equality over the signature $\<0,1,+\>$.
We use the following abbreviations, where $s$ and $t$ denote arbitrary terms over the signature $\<0,1,+\>$:
\begin{itemize}
	\item $s \neq t$ abbreviates $\neg (s = t)$,
	\item $s \leq t$ abbreviates $\exists z.\, s + z = t$,
	\item $s < t$ abbreviates $s+1 \leq t$,
	\item for any integer $k\geq 1$ we use the constant $k$ as abbreviation for the sum $1 + \ldots + 1$ with $k$ summands,
	\item for any integer $k\geq 1$ and any variable $x$ we write $kx$ to abbreviate $x + \ldots + x$ with $k$ summands.
\end{itemize}
For notational convenience, we shall also use the relation symbols $\leq, <$ in their symmetric variants $\geq$ and $>$, respectively.
We follow the convention that negation binds strongest, that conjunction binds stronger than disjunction, and that all of the aforementioned bind stronger than implication. The scope of quantifiers shall stretch as far to the right as possible.

The \emph{universal fragment of Presburger arithmetic} confines the language of Presburger arithmetic to sentences in prenex normal form in which universal quantification is allowed but existential quantification may not occur. The abbreviations $s \leq t$ and $s < t$ are exempt from the rule, i.e.\ we pretend that they do not stand for a formula that contains a quantifier.

This exemption does not constitute a serious weakening of the restriction to universal quantification, because any atom $s \leq t$ can be replaced with an atom $\neg (t < s)$, which is equivalent to $\forall  z.\, \neg (t+1+z = s)$.
For instance, the sentence 
	\[ \varphi \;:=\; \forall x y.\, x = y \,\longrightarrow\, x \leq y \]
belongs to the universal fragment of Presburger arithmetic, although it is actually a short version of 
	$\forall x y.\, x = y \,\longrightarrow\, \exists z.\, x + z = y$.
However, $\varphi$ is equivalent to 
	$\forall x y.\, x = y \,\longrightarrow\,  \neg (y < x)$,
which stands for 
	$\forall x y.\, x = y \,\longrightarrow\,  \neg (\exists z.\, y+1 + z = x)$
and is thus equivalent to 
	\[ \forall x y z.\, x = y  \,\longrightarrow\, y+1 + z \neq x ~.\]

\subsection{Minsky's two-counter machines}

Minsky has introduced the two-counter machine as a Turing-complete model of computation (Theorem 14.1-1 in \cite{Minsky1967}).
We shall only briefly recap the basic architecture of this kind of computing device.

A \emph{two-counter machine $\cM$} consists of two counters $C_1, C_2$ and a finite program whose lines are labeled with integers $0, \ldots, K$.
Each program line contains one of five possible instructions with the following meaning:
\begin{center}
	\begin{tabular}{cl}
		\texttt{inc($C_1$)}			&	increment counter $C_1$ and proceed with the next instruction; \\[1ex]
		\texttt{inc($C_2$)}			&	increment counter $C_2$ and proceed with the next instruction; \\[1ex]
		\texttt{test\&dec($C_1$,$\ell$)}	&	if $C_1 > 0$ then decrement $C_1$ and proceed with the next instruction, \\
							&	otherwise proceed with instruction $\ell$ and leave the counters unchanged; \\[1ex]
		\texttt{test\&dec($C_2$,$\ell$)}	&	if $C_2 > 0$ then decrement $C_2$ and proceed with the next instruction, \\
							&	otherwise proceed with instruction $\ell$ and leave the counters unchanged; \\[1ex]
		\texttt{halt}				&	halt the computation.
	\end{tabular}	
\end{center}	
We tacitly assume that the last program line of any two-counter machine contains the \texttt{halt} instruction.
In the initial state of a given two-counter machine the input is stored in the two counters.
The computation of the machine starts at the first program line, labeled $0$.

Notice that the described machine model leads to deterministic computation processes.
Since the described machine model is strong enough to simulate any deterministic Turing machine, the halting problem for two-counter machines is undecidable.
\begin{proposition}[corollary of Theorem~14.1-1 from \cite{Minsky1967}]\label{proposition:UndecidabilityTwoCounterMachines}
	It is impossible to devise an algorithm that is able to decide for every two-counter machine $\cM$ and every input $\<m,n\> \in \Nat \times \Nat$ whether $\cM$ ever reaches a program line containing the \texttt{halt} instruction when started on $\<m,n\>$.
\end{proposition}

\section{Encoding a deterministic two-counter machine}\label{section:Encoding}

Since validity of Presburger arithmetic sentences is decidable\footnote{Regarding sentences, i.e.\ closed formulas, of Presburger arithmetic, validity and satisfiability coincide. The reason is that the domain is fixed to the nonnegative integers and all language elements in the underlying signature $\<0,1,+\>$ have a fixed interpretation. As soon as we add uninterpreted operations or relations, the two notions differ. In this latter case, we shall consider decidability of the satisfiability problem rather than of the validity problem.}, we need some additional  language element in order to encode the computations of a two-counter machine on a given input.
It turns out that it is sufficient to add an uninterpreted unary predicate symbol $P$ to the underlying signature and thus consider first-order sentences over the extended signature $\<0,1,+,P\>$.
As soon as we have constructed a sentence $\varphi$ that encodes a given machines $\cM$ together with a given input pair $\<m,n\>$, we are interested in the (un)satisfiability of $\varphi$.
Hence, we pose the question: Is there an interpretation $\cI$ with $P^\cI \subseteq \Nat$ such that $\cI \models \varphi$, or is there no such interpretation?

\subsection{Informal description of the encoding}

Since any interpretation $P^\cI$ of the predicate symbol $P$ is a subset of the natural numbers, we can view $P^\cI$ as an infinite sequence of bits $b_0 b_1 b_2 \ldots$, where for every $n \in \Nat$ we have
	\[ b_n := \begin{cases} 0 & \text{if $n \not\in P^\cI$},\\  1 & \text{if $n \in P^\cI$}. \end{cases} \]	
Given a two-counter machine $\cM$ with $K + 1$ program lines, labeled $0, \ldots, K$, and two input values $m, n$, we shall encode all the configurations that occur during the run of $\cM$ when started on input $\<m,n\>$.
One such configuration consists of the address of the program line that is to be executed in the current step, the value of the first counter $C_1$, and the value of the second counter $C_2$.
We divide the bit sequence $P^\cI$ into chunks of growing length, each delimited by the bit sequence $001011$.
Such a chunk is divided into three subchunks, using the bit sequence $0011$ as a delimiter.
The first subchunk holds the current program line encoded in unary.
The second and third subchunk store the current values of the counters $C_1, C_2$, respectively, also encoded in unary notation.
Hence, every chunk has the form
	\[ \underbrace{001011}_{\text{\scriptsize\begin{tabular}{c}left de-\\ limiter \end{tabular}}}\hspace{-1ex}
		1^\ell
		0 \ldots 0\hspace{-2ex}
		\underbrace{0011}_{\text{\scriptsize\begin{tabular}{c}first sub-\\ delimiter \end{tabular}}}\hspace{-3ex}
		1^{c_1}
		0 \ldots 0\hspace{-1.5ex}
		\underbrace{0011}_{\text{\scriptsize\begin{tabular}{c}second\\ subde- \\ limiter \end{tabular}}}\hspace{-2ex}
		1^{c_2}
		0 \ldots 0
		~,\]
where $\ell$ is the address of the program line to be executed, $c_1$ is the value currently stored in counter $C_1$, and $c_2$ is the value currently stored in counter $C_2$.
The subsequences $1^\ell$, $1^{c_1}$ and $1^{c_2}$ are followed by blocks of zeros that fill up the gap before the next $0011$-delimiter (indicating the start of the subsequent subchunk) or the next $001011$-delimiter (indicating the beginning of the successor configuration).

The length of each chunk and its subchunks increases with the number of computation steps that have already been performed.
This makes sure that there is always enough space available to store the current counter values, which may thus become arbitrarily large.
Of course, we have to provide sufficient space in the beginning such that the address of any program line and the initial counter values $m$ and $n$ may be stored.
In order to achieve this, we define the constant $d := \max\{K, m, n\} + 6$ and require that the leftmost chunk starts at position $d$, i.e.\ there is a $001011$-delimiter starting at position $d$ but none starting left of $d$.
The first three subchunks have length $d$ each. 
Thus, the second chunk starts at position $4c$. The subchunks of the second chunk, however, have a length of $4 d$.\footnote{Technically, a length of $d+1$ for the subchunks of the second chunk would suffice. After all, the value of a counter can increase by at most one  in a single computation step. However, we have chosen to increase the length in an exponential fashion rather than a linear one, in order to keep the encoding simple.} Hence, the total length of the second chunk is $12 d$.
This scheme continues indefinitely, i.e.\ the starting points of the chunks in the bit sequence are $d, 4 d, 16 d, 64 d, 256 d$, and so on.
Consequently, all the chunks are large enough to store all possibly occurring counter values, as these can increase by at most one in every step of the computation.

The following figure illustrates the structure of a single chunk in the sequence, starting at position $x$.
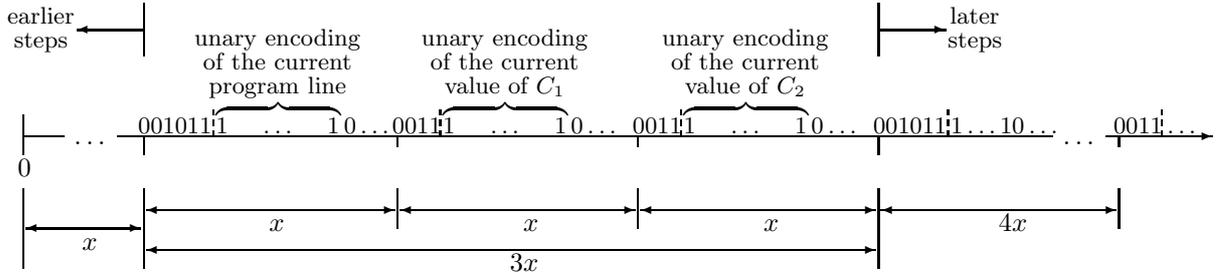
\begin{center}
	\begin{picture}(450,100)
		\put(-2,35){\mbox{$0$}}
		
		\put(0,50){\line(1,0){15}}
		\put(19,47){\mbox{$\ldots$}}
		\put(35,50){\line(1,0){350}}
		\put(389,47){\mbox{$\ldots$}}
		\put(405,50){\vector(1,0){40}}
		
		\put(0,43){\line(0,1){15}}
		\put(45,43){\line(0,1){7}}
		\put(140,46){\line(0,1){4}}
		\put(230,46){\line(0,1){4}}
		\put(320,43){\line(0,1){7}}
		\put(410,46){\line(0,1){4}}
		
		\put(43,51){\mbox{\small$001011$}}
		\put(138,51){\mbox{\small$0011$}}
		\put(228,51){\mbox{\small$0011$}}
		\put(318,51){\mbox{\small$001011$}}
		\put(408,51){\mbox{\small$0011\ldots$}}
		
		\multiput(71,50)(0,4){3}{\line(0,1){2}}
		\multiput(156,50)(0,4){3}{\line(0,1){2}}
		\multiput(246,50)(0,4){3}{\line(0,1){2}}
		\multiput(346,50)(0,4){3}{\line(0,1){2}}
		\multiput(426,50)(0,4){3}{\line(0,1){2}}
		
		\put(72,51){\mbox{\small$\overbrace{1\hspace{3ex}\ldots\hspace{3ex} 1}0\ldots$}}
		\put(157,51){\mbox{\small$\overbrace{1\hspace{3ex}\ldots\hspace{3ex} 1}0\ldots$}}
		\put(247,51){\mbox{\small$\overbrace{1\hspace{3ex}\ldots\hspace{3ex} 1}0\ldots$}}
		\put(347,51){\mbox{\small$1\ldots 10\ldots$}}
		
		\put(58,75){\small\begin{tabular}{c}unary encoding \\[-0.5ex] of the current\\[-0.5ex] program line \end{tabular}}
		\put(143,75){\small\begin{tabular}{c}unary encoding \\[-0.5ex] of the current\\[-0.5ex] value of $C_1$ \end{tabular}}
		\put(233,75){\small\begin{tabular}{c}unary encoding \\[-0.5ex] of the current\\[-0.5ex] value of $C_2$ \end{tabular}}
		
		\put(45, 80){\line(0,1){20}}
		\put(45,90){\vector(-1,0){25}}
		\put(-12,88){\small\begin{tabular}{c} earlier\\[-0.5ex] steps \end{tabular}}
		\put(320, 80){\line(0,1){20}}
		\put(320,90){\vector(1,0){25}}
		\put(340,88){\small\begin{tabular}{c} later\\[-0.5ex] steps \end{tabular}}
		
		\put(0,0){\line(0,1){30}}
		\put(45,0){\line(0,1){30}}
		\put(140,15){\line(0,1){15}}
		\put(230,15){\line(0,1){15}}
		\put(320,0){\line(0,1){30}}
		\put(410,15){\line(0,1){15}}
		
		\put(1,15){\vector(1,0){43}}
		\put(44,15){\vector(-1,0){43}}
		\put(22,7){\mbox{$x$}}
		\put(46,7){\vector(1,0){273}}
		\put(319,7){\vector(-1,0){273}}
		\put(182,-1){\mbox{$3x$}}
		\put(46,22){\vector(1,0){93}}
		\put(139,22){\vector(-1,0){93}}
		\put(92,14){\mbox{$x$}}
		\put(141,22){\vector(1,0){88}}
		\put(229,22){\vector(-1,0){88}}
		\put(187,14){\mbox{$x$}}
		\put(231,22){\vector(1,0){88}}
		\put(319,22){\vector(-1,0){88}}
		\put(277,14){\mbox{$x$}}
		\put(321,22){\vector(1,0){88}}
		\put(409,22){\vector(-1,0){88}}
		\put(365,14){\mbox{$4x$}}
	\end{picture}		
\end{center}

\subsection{Formal encoding of two-counter machine computations}\label{section:FirstEncoding}

Recall that we assume to be given a two-counter machine $\cM$ with $K + 1$ program lines, labeled $0, \ldots, K$, and two input values $m$ and $n$.
We use the following abbreviations for arbitrary terms $t$:
\begin{align*}
	\phiChunk(t) &:= \neg P(t) \;\wedge\; \neg P(t+1) \;\wedge\; P(t+2) \;\wedge\; \neg P(t+3) \;\wedge\; P(t+4) \;\wedge\; P(t+5) \\
	\phiBreak(t) &:= \neg P(t) \;\wedge\; \neg P(t+1) \;\wedge\; P(t+2) \;\wedge\; P(t+3) \\
	\phiStart(t) &:= \neg P(t) \;\wedge\; P(t+1) \\
	\phiEnd(t) &:= P(t) \;\wedge\; \neg P(t+1) \\[1ex]
	\chi_j(t) &:= \phiEnd(t+5+j)	\hspace{10ex} \text{for $j = 0, \ldots, K$}
\end{align*}

First of all, we set up the general structure of the predicate $P$.
Let $d$ denote the integer with the value $d := \max\{K+6, m+4, n+4\}$. 
We use $d$ as the starting point of our encoding.
\begin{align}
	\varphi_1 :=\;\; \label{formula:I}
	&\phiChunk(d) \\
	\label{formula:II}
	&\wedge\; \bigl( \forall x.\; x < d \;\longrightarrow\; \neg P(x) \bigr) \\
	\label{formula:III}
	&\wedge\; \bigl( \forall x.\; \phiChunk(x) \;\;\longrightarrow\;\; \phiBreak(2x) \;\wedge\; \phiBreak(3x) \;\wedge\; \phiChunk(4x) \bigr) \\
	\label{formula:IV}
	&\wedge\; \bigl( \forall x y.\; \phiChunk(x) \;\wedge\; \phiChunk(y) \;\wedge\; x \leq y \;\wedge\; y < 4x \;\;\longrightarrow\;\; x=y \bigr) \\
	\label{formula:V}
	&\wedge\; \bigl( \forall x y.\; \phiChunk(x) \;\wedge\; \phiBreak(y) \;\wedge\; x \leq y \;\;\longrightarrow\;\; y \geq 2x \bigr) \\
	\label{formula:VI}
	&\wedge\; \bigl( \forall x y.\; \phiChunk(x) \;\wedge\; \phiBreak(y) \;\wedge\; 2x < y \;\;\longrightarrow\;\; y \geq 3x \bigr) \\
	\label{formula:VII}
	&\wedge\; \bigl( \forall x y.\; \phiChunk(x) \;\wedge\; \phiBreak(y) \;\wedge\; 3x < y \;\;\longrightarrow\;\; y \geq 4x \bigr) \\
	\label{formula:VIII}
	&\wedge\; \bigl( \forall x y u.\; \phiChunk(x) \;\wedge\; \phiStart(y) \;\wedge\; x+5 < y \;\wedge\; y < 4x \;\wedge\; u + 1 = y \;\;\longrightarrow\;\; \phiBreak(u) \bigr)
\end{align}
Formula~(\ref{formula:I}) sets the first $001011$-delimiter at position $d$ and Formula~(\ref{formula:II}) ensures that this is indeed the leftmost such delimiter.
Formula~(\ref{formula:III}) sets up all the other delimiters and Formulas~(\ref{formula:IV}) to (\ref{formula:VII}) guarantee that there are no spurious delimiters in between them.
Formula~(\ref{formula:VIII}) stipulates that every $01$ substring is part of one of the delimiters, i.e.\ there cannot be a substring $01$ that lies outside of a $001011$- or $0011$-delimiter.
This does also entail that between one delimiter ($001011$ or $0011$) and the subsequent one there is exactly one substring $10$, possibly overlapping with the last or first bit of one of the delimiters. Hence, this substring uniquely marks the end of the number encoded in the respective subchunk.

There is one peculiarity in Formula~(\ref{formula:VIII}) that is worth noticing, namely, the role of the variable $u$.
We need to introduce it to facilitate the formulation of the term $y -1$, since the signature of Presburger arithmetic does not contain the minus operation.
Hence, informally, Formula~(\ref{formula:VIII}) stands for
	\[ \forall x y.\; \phiChunk(x) \;\wedge\; \phiStart(y) \;\wedge\; x+5 < y \;\wedge\; y < 4x \;\;\longrightarrow\;\; \phiBreak(y-1) ~.\]
We will use this pattern again later, when we shall encode the decrement operation for counters.

The following formula sets the initial values of the counters. Moreover, it sets the initial program line, which we assume to be the very first one:
	\[ \varphi_2 :=\;\; \chi_0(d) \;\;\wedge\;\; \phiEnd(2d+3+m) \;\;\wedge\;\; \phiEnd(3d+3+n) ~.\]
Regarding the encoding of program lines, we have to enforce that the current program line never exceeds $K$. 
This is easily done with the formula
	\[ \varphi_3 :=\;\;  \forall x y.\; \phiChunk(x) \;\;\wedge\;\; \phiEnd(y) \;\;\wedge\;\; x+5 \leq y \;\;\wedge\;\; y \leq 2x \;\;\longrightarrow\;\; y \leq x + 5 + K ~.\]
The previous formulas already ensure that exactly one address of a program line is encoded. 

Next we encode the control flow of $\cM$. 
We assume that the following instructions occur in program line $j$ for some $j \in \{0, \ldots, K\}$.
\begin{description}
	\item Encoding of the instruction $j:\; \texttt{inc($C_1$)}$:
		\begin{align*} 
			\forall x y z.\; &\phiChunk(x) \;\wedge\; 2x \leq y \;\wedge\; y \leq 3x \;\wedge\; \phiEnd(y) \;\wedge\; 3x \leq z \;\wedge\; z \leq 4x \;\wedge\; \phiEnd(z) \;\wedge\; \chi_j(x) \\
			&\;\;\longrightarrow\;\; \phiEnd(6x+y+1) \;\wedge\; \phiEnd(9x + z) \;\wedge\; \chi_{j+1}(4x)
		\end{align*}
		The subfomula $\phiChunk(x)$ in the premise of the implication states that the chunk encoding the currently regarded configuration starts at position $x$.
		The other preconditions make clear that $y$ and $z$ correspond to the positions at which we find $10$-substrings in the two subchunks storing the current counter values:
			\[ \begin{array}{l@{}l@{}l@{}l@{}l}
				\,x	&	2x	&	y	&	3x	&	z
				\\
				\:\downarrow	&	\:\downarrow	&	\downarrow	&	\:\downarrow	&	\downarrow
				\\
				\underbrace{001011}_{\text{\scriptsize\begin{tabular}{c}left de-\\ limiter \end{tabular}}}\hspace{-1ex}
				1^\ell
				0 \ldots 0
				&
				\hspace{-2ex}\underbrace{0011}_{\text{\scriptsize\begin{tabular}{c}first sub-\\ delimiter \end{tabular}}}\hspace{-3ex}
				1^{c_1-1}
				&
				10 \ldots 0
				&
				\hspace{-1ex}\underbrace{0011}_{\text{\scriptsize\begin{tabular}{c}second\\ subde- \\ limiter \end{tabular}}}\hspace{-2ex}
				1^{c_2-1}
				&
				10 \ldots 0
			\end{array}
			\]
		Hence, $C_1$ and $C_2$ currently hold the values $c_1 = y - 2x - 3$ and $c_2 = z - 3x - 3$, respectively.
		Since the subsequent chunk starts at position $4x$ and its second and third subchunks start at positions $8x$ and $12x$, respectively, we know that there must be one $10$-substring at position $8x+3+c_1+1 = 6x+y+1$---the first counter is incremented by $1$---and one at position $12x+3+c_2 = 9x+z$---the value of the second counter remains unchanged.
		Moreover, the machine currently executes program line $j$ and is to continue at program line $j+1$.
		Therefore, we put the formula $\chi_j(x)$ in the premise and the formula $\chi_{j+1}(4x)$ into the consequent of the implication.
		
	\item Encoding of the instruction $j:\; \texttt{inc($C_2$)}$:
		\begin{align*} 
			\forall x y z.\; &\phiChunk(x) \;\wedge\; 2x \leq y \;\wedge\; y \leq 3x \;\wedge\; \phiEnd(y) \;\wedge\; 3x \leq z \;\wedge\; z \leq 4x \;\wedge\; \phiEnd(z) \;\wedge\; \chi_j(x) \\
			&\;\;\longrightarrow\;\; \phiEnd(6x+y) \;\wedge\; \phiEnd(9x+z+1) \;\wedge\; \chi_{j+1}(4x)
		\end{align*}				

	\item Encoding of the instruction $j:\; \texttt{test\&dec($C_1$,$\ell$)}$:
		\begin{description}
			\item The case of $C_1$ storing $0$:
				\begin{align*} 
					\forall x y z.\; &\phiChunk(x) \;\wedge\; 2x \leq y \;\wedge\; y \leq 3x \;\wedge\; \phiEnd(y) \;\wedge\; 3x \leq z \;\wedge\; z \leq 4x \;\wedge\; \phiEnd(z) \\
					&\wedge\; \chi_j(x) \;\wedge\; y = 2x+3 \\
					&\;\;\longrightarrow\;\; \phiEnd(6x+y) \;\wedge\; \phiEnd(9x+z) \;\wedge\; \chi_{\ell}(4x)
				\end{align*}
				The condition $y = 2x+3$ ensures that the first counter stores the value $0$.
		
			\item The case of $C_1$ storing a value greater than $0$:
				\begin{align*} 
					\forall x y z u.\; &\phiChunk(x) \;\wedge\; 2x \leq y \;\wedge\; y \leq 3x \;\wedge\; \phiEnd(y) \;\wedge\; 3x \leq z \;\wedge\; z \leq 4x \;\wedge\; \phiEnd(z) \\
					&\wedge\; \chi_j(x) \;\wedge\; y > 2x+3 \;\wedge\; u+1 = 6x+y \\
					&\;\;\longrightarrow\;\; \phiEnd(u) \;\wedge\; \phiEnd(9x+z) \;\wedge\; \chi_{j+1}(4x)
				\end{align*}				
				The condition $y > 2x+3$ ensures that the first counter stores a value strictly greater than $0$.
				Notice that $u$ stands for the term $6x+y-1$ and thus facilitates the decrement operation.
	\end{description}	

	\item Encoding of the instruction $j:\; \texttt{test\&dec($C_2$,$\ell$)}$:
		\begin{description}
			\item The case of $C_2$ storing $0$:
				\begin{align*} 
					\forall x y z.\; &\phiChunk(x) \;\wedge\; 2x \leq y \;\wedge\; y \leq 3x \;\wedge\; \phiEnd(y) \;\wedge\; 3x \leq z \;\wedge\; z \leq 4x \;\wedge\; \phiEnd(z) \\
					&\wedge\; \chi_j(x) \;\wedge\;z = 3x+3 \\
					&\;\;\longrightarrow\;\; \phiEnd(6x+y) \;\wedge\; \phiEnd(9x+z) \;\wedge\; \chi_{\ell}(4x)
				\end{align*}

			\item The case of $C_2$ storing a value greater than $0$:
				\begin{align*} 
					\forall x y z u.\; &\phiChunk(x) \;\wedge\; 2x \leq y \;\wedge\; y \leq 3x \;\wedge\; \phiEnd(y) \;\wedge\; 3x \leq z \;\wedge\; z \leq 4x \;\wedge\; \phiEnd(z) \\
					&\wedge\; \chi_j(x) \;\wedge\; z > 3x+3 \;\wedge\; u+1 = 9x+z \\
					&\;\;\longrightarrow\;\; \phiEnd(6x+y) \;\wedge\; \phiEnd(u) \;\wedge\; \chi_{j+1}(4x)
				\end{align*}				
	\end{description}	
		
	\item Encoding of the instruction $j:\; \texttt{halt}$:
		\begin{align*} 
			\forall x y z.\; &\phiChunk(x) \;\wedge\; 2x \leq y \;\wedge\; y \leq 3x \;\wedge\; \phiEnd(y) \;\wedge\; 3x \leq z \;\wedge\; z \leq 4x \;\wedge\; \phiEnd(z) \;\wedge\; \chi_j(x) \\
			&\;\;\longrightarrow\;\; \phiEnd(6x+y) \;\wedge\; \phiEnd(9x+z) \;\wedge\; \chi_{K}(4x)
		\end{align*}				
		The consequent of the implication ensures that the counters remain unchanged and that the computation continues at program line $K$.
		Since we assumed the $K$-th program line to contain the instruction \texttt{halt}, the rest of the bit sequence will repeat the same chunk structure again and again, as the counter values will remain unchanged and the encoded program line will also repeat indefinitely.		
\end{description}
		
Finally, we pose the central question concerning the halting behavior of the machine: Does the machine ever reach a program line containing the \texttt{halt} instruction?
The question is posed as a requirement in a negative fashion:
	\[ \varphi_4 :=\;\; \forall x.\; \phiChunk(x) \;\;\longrightarrow\;\;\neg \chi_K(x) ~.\]
Technically speaking, we require that the machine never reaches the $K$-th program line. Due to the encoding of the \texttt{halt} instruction in arbitrary program lines $j$, it is clear that whenever any program line containing \texttt{halt} is reached the $K$-th program line is reached in the subsequent step. Hence, the above formula is satisfied if and only if the machine will never reach the instruction \texttt{halt} when started on the given input.
		
\begin{lemma}\label{lemma:SimulationOfComputations}
	The two-counter machine $\cM$ with $K+1$ program lines, labeled $0, \ldots, K$, started on input $\<m,n\>$ eventually reaches a program line containing the instruction \texttt{halt} if and only if the described encoding of $\cM$ in Presburger arithmetic with an additional uninterpreted unary predicate symbol is unsatisfiable.
\end{lemma}		
\begin{proof}[Proof sketch]
	We first observe the following technical properties of every interpretation $\cI$ with $\cI \models \varphi_1$.
	\begin{enumerate}[label=(\alph{*}), ref=(\alph{*})]
		\item For every integer $r \in \Nat$ we have $\cI, [x \Mapsto r] \models \phiChunk(x)$ if and only if $r = 4^i d$ for some $i \in \Nat$.
		\item For every integer $r \in \Nat$ we have $\cI, [x \Mapsto r] \models \phiBreak(x)$ if and only if $r = 2\cdot4^i d$ or $r = 3\cdot4^i d$ for some $i \in \Nat$.
		\item For every integer $r \in \Nat$ we have $\cI, [x \Mapsto r] \models \phiStart(x)$ if and only if 
			\[ r \in \bigcup_{i \in \Nat} \bigl\{ 4^i d + 1, 4^i d + 3, 2\cdot 4^i d + 1, 3\cdot 4^i d + 1 \bigr\} ~.\]
		\item Suppose there are integers $i,r,q \in \Nat$ such that $4^i d + 5 \leq r, q < 2\cdot 4^i d$. If we have $\cI, [x \Mapsto r] \models \phiEnd(x)$ and $\cI, [x \Mapsto q] \models \phiEnd(x)$, then it follows that $r = q$. 
		\item Suppose there are integers $i,r,q \in \Nat$ such that $2\cdot 4^i d + 3 \leq r, q < 3\cdot 4^i d$. If we have $\cI, [x \Mapsto r] \models \phiEnd(x)$ and $\cI, [x \Mapsto q] \models \phiEnd(x)$, then it follows that $r = q$. 
		\item Suppose there are integers $i,r,q \in \Nat$ such that $3\cdot 4^i d + 3 \leq r, q < 4\cdot 4^i d$. If we have $\cI, [x \Mapsto r] \models \phiEnd(x)$ and $\cI, [x \Mapsto q] \models \phiEnd(x)$, then it follows that $r = q$. 
		\item For every integer $i \in \Nat$ there are integers $r_1, r_2, r_3 \in \Nat$ such that 
			\begin{itemize}
				\item $4^i d + 5 \leq r_1 < 2\cdot 4^i d$ and $\cI, [x \Mapsto r_1] \models \phiEnd(x)$,
				\item $2\cdot 4^i d + 3 \leq r_2 < 3\cdot 4^i d$ and $\cI, [x \Mapsto r_2] \models \phiEnd(x)$, and
				\item $3\cdot 4^i d + 3 \leq r_3 < 4\cdot 4^i d$ and $\cI, [x \Mapsto r_3] \models \phiEnd(x)$.
			\end{itemize}
	\end{enumerate}
	Due to the above observations, it is clear that any model $\cI$ of $\varphi_1$ interprets $P$ in such a way that it uniquely represents an infinite sequence of triples of nonnegative integers encoded in unary, just as we have described it earlier.
	If, in addition, $\cI$ satisfies $\varphi_2$ and $\varphi_3$, then the first triple of the sequence has the form $\<0,m,n\>$ and the first component of every triple in the sequence does not exceed $K$.
	
	Given the program of $\cM$, we denote by $\varphi_\cM$ the sentence that encodes $\cM$'s program in accordance with the described formula schemes.
	Hence, for any model $\cI \models \varphi_1 \wedge \varphi_2 \wedge \varphi_3 \wedge \varphi_\cM$ the interpretation $P^\cI$ of $P$ does not only represent a sequence of triples of integers but also establishes relations between the triples in the sequence, such that they mimic the operations that $\cM$ would perform on its counters in accordance with the instructions in its program.
	The only technical difference is that whenever $\cM$ enters a configuration $\<\ell, c_1, c_2\>$ such that instruction $\ell$ in $\cM$'s program is \texttt{halt}, then all later configurations have the form $\<K, c_1, c_2\>$.
	All in all, $P^\cI$ is a faithful encoding of some run of $\cM$ starting from the input $\<m,n\>$.
	
	On the other hand, since $\cM$ is deterministic, there is a unique sequence 
		\[ \tau := \<0,m,n\>\<\ell_1, c_{1,1}, c_{2,1}\>\<\ell_2, c_{1,2}, c_{2,2}\>\<\ell_3, c_{1,3}, c_{2,3}\>\ldots \]
	of configurations that represents the \emph{run of $\cM$ started on input $\<m,n\>$}. If $\tau$ is finite and thus contains a halting configuration $\<\ell, c_1, c_2\>$ as its last triple, we concatenate the infinite sequence $\<K,c_1,c_2\>\<K,c_1,c_2\> \ldots$ and thus obtain an infinite sequence again.
	This infinite sequence (be it originally infinite or made so artificially) can be translated into an interpretation $\cI_\tau$ such that $\cI_\tau \models \varphi_1 \wedge \varphi_2 \wedge \varphi_3 \wedge \varphi_\cM$.
	
	So far, we have seen that $\varphi_1 \wedge \varphi_2 \wedge \varphi_3 \wedge \varphi_\cM$ is satisfiable and that every model represents the unique run $\tau$ of $\cM$ started on input $\<m,n\>$.
	Clearly, we now observe for any model $\cI \models \varphi_1 \wedge \varphi_2 \wedge \varphi_3 \wedge \varphi_\cM$ that $\cI \models \varphi_4$ holds if and only if $\tau$ \emph{does not} contain a triple $\<K, c_1, c_2\>$ for any $c_1, c_2 \in \Nat$. Hence, $\varphi_1 \wedge \varphi_2 \wedge \varphi_3 \wedge \varphi_\cM \wedge \varphi_4$ is unsatisfiable if and only if $\cM$ reaches the \texttt{halt} instruction when started on the input $\<m,n\>$.
\end{proof}
Together with the fact that the halting problem for two-counter machines is undecidable (cf.\ Proposition~\ref{proposition:UndecidabilityTwoCounterMachines}), we have shown the following theorem.
\begin{theorem}\label{theorem:UndecidabilityOverIntegers}
	(Un)satisfiability of the universal fragment of Presburger arithmetic with a single additional uninterpreted unary predicate symbol is undecidable.
\end{theorem}


An alternative proof of Lemma~\ref{lemma:SimulationOfComputations} uses \emph{hierarchic superposition}~\cite{Bachmair1994b,Althaus2009,Kruglov2012,Baumgartner2013}.
Transforming the two-counter machine encoding $\varphi_1 \wedge \varphi_2 \wedge \varphi_3 \wedge \varphi_\cM \wedge \varphi_4$ 
into \emph{conjunctive normal form (CNF)} results in a first-order clause set $\Phi$ over Presburger arithmetic without any uninterpreted function or
constant symbols. The only uninterpreted symbol in $\Phi$ is the predicate $P$. Therefore, hierarchic superposition is \emph{refutationally complete} for this clause set.
That is, $\Phi$ unsatisfiable if and only if hierarchic superposition derives a contradiction, the empty clause, out of the clauses
in $\Phi$. Hierarchic superposition first resolves away all $P$ literals from a clause and in case the clause only
contains arithmetic literals, it checks their unsatisfiability.

By an inductive argument,
hierarchic superposition derives exactly the states of the two-counter machine via ground literals of the form $[\neg]P(k)$, $k\in \Nat$.
Let $\PhiChunk(s)$, $\PhiEnd(s)$, and $\Chi_j(s)$ be the sets of clauses---unit clauses in this case---that correspond to the formulas $\phiChunk(s)$, $\phiEnd(s)$, $\chi_j(s)$, respectively, for any term $s$.

Suppose that the ground clauses in $\PhiChunk(k)$, $\Chi_j(k)$, $\PhiEnd(2k+3+c_1)$, and
$\PhiEnd(3k+3+c_2)$ (with $2k+3+c_1 < 3k$ and $3k+3+c_2 < 4k$) have been derived already. They represent the two-counter machine at program line $j$ with counter values $c_1$ and $c_2$. Without loss of generality, we
assume that the instruction of the machine at line $j$ is an increment on the second counter. The other operations are argued analogously.
Consider the clauses that result from the formula
\begin{align*} 
  \forall x y z.\; &\phiChunk(x) \;\wedge\; 2x \leq y \;\wedge\; y \leq 3x \;\wedge\; \phiEnd(y) \;\wedge\; 3x \leq z \;\wedge\; z \leq 4x \;\wedge\; \phiEnd(z) \;\wedge\; \chi_j(x) \\
		&\;\;\longrightarrow\;\; \phiEnd(6x+y) \;\wedge\; \phiEnd(9x+z+1) \;\wedge\; \chi_{j+1}(4x) ~.
\end{align*}
All literals in $\phiChunk(x)$, $\phiEnd(y)$, $\phiEnd(z)$, and $\chi_j(x)$ can be resolved away via superposition and the substitution $\sigma = \{x\mapsto k, y\mapsto 2k+3+c_1, z\mapsto 3k+3+c_2\}$, thus
generating the new ground clauses in $\PhiEnd(6x+y)\sigma$, $\PhiEnd(9x+z+1)\sigma$, $\Chi_{j+1}(4x)\sigma$, which represent the next state of the two-counter machine.
The underlying strategy always selects all negative literals resulting out of $\phiChunk(x)$, $\phiEnd(y)$, $\phiEnd(z)$, and $\chi_j(x)$. Resolving those away with the respective
ground unit literals, fixes already the values for $x$, $y$, and $z$. Thus all other $P$ literals resulting from the delimiter encoding can then be reduced away by
subsumption resolution.
The ground literals in $\PhiChunk(4x)\sigma$ are generated by this strategy with the clauses resulting from (\ref{formula:III}). This finishes the proof that hierarchic superposition
generates the states of the two-counter machine by the derivation of respective $P$ ground literals. 

However, it remains to be shown that there are no derivations with other clauses that might lead to a contradiction.
Any clause resulting from the encoding of a two-counter machine instruction contains the literals in $\Chi_\ell(x)$ for the
respective line of the program $\ell$. Now if the encoding of the instruction of line $\ell$ is resolved with ground $P$ literals
representing a line $j\neq \ell$ via the above selection strategy, then all resulting clauses turn into tautologies as $\chi_\ell(x)$ does not hold for the respective
positions as a result of the different delimiter bit sequences that cannot be confused.
Hence all these inferences become redundant and can be neglected. This is a consequence of the clauses resulting from (\ref{formula:I}) to (\ref{formula:VIII}).
So given ground literals $\phiChunk(k)$, $\chi_j(k)$, $\phiEnd(2k+3+c_1)$,
 $\phiEnd(3k+3+c_2)$, there is exactly one line encoding formula that can be resolved with that does not lead to tautologies:
the encoding of program line $j$.

\subsection{Reducing the number of variables}\label{section:ReducedEncoding}
	
We can formulate the encoding with at most two variables per formula, if we are willing to incorporate $\leq$ into the signature of Presburger arithmetic rather than conceiving expressions of the form $s \leq t$ as an abbreviation of $\exists z. s + z = t$. If we accept this extended signature, we still have to modify the encoding formulas a little bit.
As a matter of fact, the criterion is already met by most of the formulas in the previous subsection.
Only Subformula~(\ref{formula:VIII}) of $\varphi_1$ and the encodings of the program instructions have to be modified as follows.
\begin{description}			
	\item Modified variant of Subformula~(\ref{formula:VIII}):
		\[ \forall x y.\; \phiChunk(x) \;\wedge\; \phiStart(y+1) \;\wedge\; x+5 < y+1 \;\wedge\; y+1 < 4x \;\;\longrightarrow\;\; \phiBreak(y) \]
		In this variant of the subformula we get rid of the variable $u$ that we have introduced in order to simulate subtraction.
		We do so by using $y+1$ in the premise rather than $y$. In this way subtraction by one in the consequent amounts to leaving away the $+1$.
		We will reuse this pattern in the encodings of the decrement operation later.
	
	\item Modified encoding of the instruction $j:\; \texttt{inc($C_1$)}$:
		\begin{align*} 
			\forall x y.\; &\phiChunk(x) \;\wedge\; 2x \leq y \;\wedge\; y \leq 3x \;\wedge\; \phiEnd(y) \;\wedge\; \chi_j(x) \;\;\longrightarrow\;\; \phiEnd(6x+y+1) \;\wedge\; \chi_{j+1}(4x)\\[2ex]
			\forall x z.\; &\phiChunk(x) \;\wedge\; 3x \leq z \;\wedge\; z \leq 4x \;\wedge\; \phiEnd(z) \;\wedge\; \chi_j(x) \;\;\longrightarrow\;\; \phiEnd(9x + z)
		\end{align*}				
		For this instruction and most of the others we split the encoding formula into two parts: the first formula realizes the $y$-part of the original encoding and the second formula realizes the $z$-part.

	\item Modified encoding of the instruction $j:\; \texttt{inc($C_2$)}$:
		\begin{align*} 
			\forall x y.\; &\phiChunk(x) \;\wedge\; 2x \leq y \;\wedge\; y \leq 3x \;\wedge\; \phiEnd(y) \;\wedge\; \chi_j(x) \;\;\longrightarrow\;\; \phiEnd(6x+y) \;\wedge\; \chi_{j+1}(4x)\\[2ex]
			\forall x z.\; &\phiChunk(x) \;\wedge\; 3x \leq z \;\wedge\; z \leq 4x \;\wedge\; \phiEnd(z) \;\wedge\; \chi_j(x) \;\;\longrightarrow\;\; \phiEnd(9x+z+1)
		\end{align*}				
		
	\item Modified encoding of the instruction $j:\; \texttt{test\&dec($C_1$,$\ell$)}$:
		\begin{description}
			\item The case of $C_1$ storing $0$:
				\begin{align*} 
					\forall x z.\; &\phiChunk(x) \;\wedge\; 3x \leq z \;\wedge\; z \leq 4x \;\wedge\; \phiEnd(z) \;\wedge\; \chi_j(x) \;\wedge\; \phiEnd(2x+3) \\
					&\;\;\longrightarrow\;\; \phiEnd(8x+3) \;\wedge\; \phiEnd(9x+z) \;\wedge\; \chi_{\ell}(4x)
				\end{align*}				
				The subformula $\phiEnd(2x+3)$ in the premise ensures that the counter $C_1$ currently stores a $0$ and the subformula $\phiEnd(8x+3)$ requires that $C_1$ still stores $0$ in the next step. Notice that we do not need a variable $y$ to address the corresponding bit positions, since we can directly calculate these positions from $x$.
				
			\item The case of $C_1$ storing a value greater than $0$:
				\begin{align*} 
					\forall x y.\; &\phiChunk(x) \;\wedge\; 2x \leq y+1 \;\wedge\; y+1 \leq 3x \;\wedge\; \phiEnd(y+1) \;\wedge\; \chi_j(x) \;\wedge\; y+1 > 2x+3 \\
						&\longrightarrow\;\; \phiEnd(6x+y) \;\wedge\; \chi_{j+1}(4x)\\[2ex]
					\forall x z.\; &\phiChunk(x) \;\wedge\; 3x \leq z \;\wedge\; z \leq 4x \;\wedge\; \phiEnd(z) \;\wedge\; \chi_j(x) \;\wedge\; \neg \phiEnd(2x+3) \\
					&\;\;\longrightarrow\;\; \phiEnd(9x+z)
				\end{align*}		
				In the first formula $y+1 > 2x+3$ ensures that the value of $C_1$ is greater than zero.
				We need the variable $y$ to refer to $C_1$'s value for the decrement operation.
				In the second formula $C_1$'s exact value is not important and thus $\neg \phiEnd(2x+3)$ is sufficient for ensuring that $C_1$'s value is strictly positive.
		\end{description}	

	\item Modified encoding of the instruction $j:\; \texttt{test\&dec($C_2$,$\ell$)}$:
		\begin{description}
			\item The case of $C_2$ storing $0$:
				\begin{align*} 
					\forall x y.\; &\phiChunk(x) \;\wedge\; 2x \leq y \;\wedge\; y \leq 3x \;\wedge\; \phiEnd(y) \;\wedge\; \chi_j(x) \;\wedge\; \phiEnd(3x+3) \\
					&\;\;\longrightarrow\;\; \phiEnd(6x+y) \;\wedge\; \phiEnd(12x+3) \;\wedge\; \chi_{\ell}(4x)
				\end{align*}				

			\item The case of $C_2$ storing a value greater than $0$:
				\begin{align*} 
					\forall x y.\; &\phiChunk(x) \;\wedge\; 2x \leq y \;\wedge\; y \leq 3x \;\wedge\; \phiEnd(y) \;\wedge\; \chi_j(x) \;\wedge\; \neg \phiEnd(3x+3) \\
					&\;\;\longrightarrow\;\; \phiEnd(6x+y) \;\wedge\; \chi_{j+1}(4x)\\[2ex]
					\forall x z.\; &\phiChunk(x) \;\wedge\; 3x \leq z+1 \;\wedge\; z+1 \leq 4x \;\wedge\; \phiEnd(z+1) \;\wedge\; \chi_j(x) \;\wedge\; z+1 > 3x+3 \\
					&\longrightarrow\;\; \phiEnd(9x+z)
				\end{align*}				
		\end{description}	

	\item Modified encoding of the instruction $j:\; \texttt{halt}$:
		\begin{align*} 
			\forall x y.\; &\phiChunk(x) \;\wedge\; 2x \leq y \;\wedge\; y \leq 3x \;\wedge\; \phiEnd(y) \;\wedge\; \chi_j(x) \;\;\longrightarrow\;\; \phiEnd(6x+y) \;\wedge\; \chi_{K}(4x) \\[2ex]
			\forall x z.\; &\phiChunk(x) \;\wedge\; 3x \leq z \;\wedge\; z \leq 4x \;\wedge\; \phiEnd(z) \;\wedge\; \chi_j(x) \;\;\longrightarrow\;\; \phiEnd(9x+z)
		\end{align*}
\end{description}

\begin{theorem}\label{theorem:UndecidabilityOverIntegersWithReducedVariables}
	Extend the standard signature $\<0,1,+\>$ of Presburger arithmetic by the relation $\leq$ and an uninterpreted unary predicate symbol $P$ to obtain $\<0,1,+,\leq,P\>$.	
	(Un)satisfiability of the universal fragment of Presburger arithmetic over this extended signature is undecidable, if we allow at least two variables per clause.
\end{theorem}
To the authors' knowledge, the case of a single variable per clause is not known to be\linebreak (un)decidable.
From Downey's encoding \cite{Downey1972} it follows that one variable is sufficient, if we further extend the signature by modulo operators $\text{mod } k$ for finitely many integer constants $k$.

\subsection{Using the reals as underlying domain}

Presburger arithmetic is defined on the natural numbers and we have shown that adding a unary uninterpreted predicate symbol leads to undecidability.
It is known that validity and satisfiability over real numbers exhibits a different behavior when decidability is concerned.

In the context of the reals, we consider $\leq$ and $<$ to be part of the signature and not abbreviations for more complicated terms.
Of course, we assume that they have the standard semantics over the reals.

We can directly use the encoding that we have presented for the integers in order to show undecidability over the real domain.
The crucial point is that we have encoded the reachability of the \texttt{halt} instruction in a negative fashion. 
If the machine $\cM$ reaches a \texttt{halt} instruction, then we cannot find a model of the encoding formula set, since any interpretation that faithfully represents the run of $\cM$ on the given input must violate the condition $\neg \chi_K(n)$ for some integer $n$ for which $\phiChunk(n)$ is true. 
We have used this observation to prove Lemma~\ref{lemma:SimulationOfComputations}.
The described conflict does not vanish when we assume a larger domain.
If, on the other hand, the machine $\cM$ does not reach a \texttt{halt} instruction, then there is a model of the encoding formula set. 
In particular, there is a model in which $P$ is interpreted such that it exclusively contains integers and no reals at all.
Hence, the fact that we are dealing with an extended domain does not affect the circumstances  under which the encoding formula set is unsatisfiable or not.
Consequently, we have the following undecidability result.

\begin{theorem}\label{theorem:UndecidabilityOverReals}
	(Un)satisfiability of the universal fragment of linear arithmetic over the reals with an additional uninterpreted unary predicate symbol is undecidable.
\end{theorem}

\subsection{Unary function symbols and the Horn fragment}\label{section:UndecidabilityWithFunctionSymbols}

The uninterpreted unary predicate $P$ in our encoding of two-counter machines can be replaced by an uninterpreted unary function $f : \Nat \to \Nat$ over the natural numbers.
We simply add the assertion $\forall x.\, f(x) \leq 1$ and substitute every negative literal $\neg P(t)$ with $f(t) = 0$ and every positive literal $P(t)$ with $f(t) = 1$, where $t$ is any term. (Implicitly, we exploit the fact that $f$ is interpreted by a total function $f^\cI$ in any interpretation $\cI$.)
After this substitution, transforming the encoding formula set from Section~\ref{section:FirstEncoding} into \emph{conjunctive normal norm (CNF)} yields a clause set that is \emph{Horn}, i.e.\ every clause contains at most one positive literal.
By this line of argument we obtain the following theorem.
\begin{theorem}\label{theorem:UndecidabilityWithFunctions}
	(Un)satisfiability of the universal Horn fragment of Presburger arithmetic with a single additional uninterpreted unary function symbol is undecidable.
\end{theorem}

Over the domain of the reals, we can replace the predicate symbol $P$ in the same spirit, yet in a slightly different way.
For one thing, we add the assertion $\forall x.\; 0 \leq f(x) \;\wedge\; f(x) \leq 1$ to the encoding, which also introduces an explicit lower bound to the values of $f$.
As this assertion alone does not guarantee that in any model the image of $f : \Real \to \Real$ contains at most two values, we replace any occurrence of $\neg P(t)$ with $f(t) = 0$ and any occurrence of $P(t)$ with $f(t) > 0$.
Again, a CNF transformation leads to a Horn clause set.
\begin{theorem}\label{theorem:UndecidabilityWithFunctionsOverReals}
	(Un)satisfiability of the universal Horn fragment of linear arithmetic over the reals with a single additional uninterpreted unary function symbol is undecidable.
\end{theorem}

\subsection{On the degree of unsolvability}\label{section:UnsolvabilityDegree}

We have shown that the unsatisfiability problem of the universal fragment of Presburger arithmetic with additional uninterpreted predicate symbols is undecidable.
Next, we shall argue that this set is recursively enumerable.
In order to prove this, it suffices to give a sound calculus that, given an unsatisfiable sentence over the language in question, derives the empty clause, i.e.\ \emph{falsity}, in finitely many steps.
This property is known as \emph{refutational completeness}. 
In fact, the mentioned calculus would constitute a semi decision procedure for unsatisfiable sentences.

Indeed, \emph{hierarchic superposition} is refutationally complete for all sets of (hierarchic) clauses that are \emph{sufficiently complete}, if the background theory is \emph{compact} (cf.\ Theorem~24 in \cite{Bachmair1994b}).
In the case of the universal fragment of Presburger arithmetic without uninterpreted constant or function symbols, the two requirements are satisfied.
Sufficient completeness (cf.\ Definition~20 in \cite{Bachmair1994b}) concerns uninterpreted constant or function symbols that reach into the background sort $\Nat$.
Since we do not allow such symbols in our language, all sets of sentences are sufficiently complete.
For the same reason, the background theory (over universal sentences built from the signature $\<0,1,+\>$) is compact.
This means, every unsatisfiable set of universal first-order sentences over $\< 0,1,+ \>$ has some finite subset that is unsatisfiable.
Indeed, every unsatisfiable set of such sentences has an unsatisfiable subset that contains exactly one sentence.
Hence, the following proposition holds.
\begin{proposition}\label{proposition:UnivPresRE}
	The set of unsatisfiable sentences over the universal fragment of Presburger arithmetic with additional uninterpreted predicate symbols (not necessarily unary ones) is recursively enumerable.
\end{proposition}

From the literature on the \emph{arithmetical hierarchy} (see, e.g.\ \cite{Rogers1987, Soare1987, Odifreddi1992}) we know the following.
\begin{proposition}\label{proposition:SigmaBasics}~
	\begin{enumerate}[label=(\roman{*}), ref=(\roman{*})]
		\item\label{proposition:ReSets:One} The set $\Sigma^0_1$ captures exactly the recursively enumerable sets.
		\item\label{proposition:ReSets:Two} The set $\Pi^0_1$ captures exactly the sets whose complement is recursively enumerable.
		\item\label{proposition:ReSets:Three} The halting problem for (ordinary) Turing machines is $\Sigma^0_1$-complete.
	\end{enumerate}	
\end{proposition}
\begin{proof}
	\ref{proposition:ReSets:One} and \ref{proposition:ReSets:Two} are reformulations of Theorems~II.1.2 and IV.1.3 in \cite{Soare1987}, respectively.
	\ref{proposition:ReSets:Three} combines the following parts of \cite{Soare1987}: Definitions~I.3.1, I.4.1, I.4.5, Theorem~II.4.2 and the discussion after Definition~IV.2.1 on page 64.
\end{proof}	

Since we have completed a chain of reductions from the halting problem of Turing machines via the halting problem of two-counter machines to the unsatisfiability problem of the universal fragment of Presburger arithmetic with uninterpreted predicate symbols, we conclude $\Sigma^0_1$-completeness of the latter problem by Lemma~\ref{lemma:SimulationOfComputations} together with  Propositions~\ref{proposition:UnivPresRE} and \ref{proposition:SigmaBasics}.
\begin{theorem}\label{theorem:UnsatUnivPresSigmaZeroOneComplete}
	The set of unsatisfiable sentences over the universal fragment of Presburger arithmetic with additional uninterpreted predicate symbols is $\Sigma^0_1$-complete.
\end{theorem}
It is worth to notice that the theorem can be translated to the realm of linear arithmetic over the reals. The reason is that hierarchic superposition is also refutationally complete over the universal fragment of this language, if there are no uninterpreted constant or function symbols involved.

Since any reduction of a problem $S$ to a problem $T$ (both read as a set of G\"odel numbers) at the same time yields a reduction from $\overline{S}$ to $\overline{T}$, the complement of a $\Sigma^0_1$-complete set is complete for $\Pi^0_1$. Hence, Theorem~\ref{theorem:UnsatUnivPresSigmaZeroOneComplete} entails $\Pi^0_1$-completeness of the set of satisfiable sentences over the same language.

There are strong ties between (un)satisfiability in the universal fragment of the language we consider and (in)validity in the dual language, the existential fragment.
The bottom line is that the obtained completeness results can be transfered to the corresponding (in)validity problems.
The overall situation is depicted in Table~\ref{table:SatAndVal}.

For the sake of completeness, we briefly discuss (un)satisfiability for the existential fragment.
Kruglov and Weidenbach \cite{Kruglov2012} have presented a general result regarding the satisfiability problem for hierarchic clause sets that are ground.
More precisely, they have devised a decision procedure for the problem that is based on a hierarchic superposition calculus.
\begin{proposition}[corollary of Theorem~23 from \cite{Kruglov2012}]
	Satisfiability of the existential fragment of Presburger arithmetic with additional uninterpreted predicate symbols is decidable.
\end{proposition}
With this knowledge we can complete the overview in Table~\ref{table:SatAndVal} and thus reveal the full picture of where the (un)satisfiability and (in)validity problems of the universal and existential fragments of Presburger arithmetic augmented with uninterpreted predicate symbols reside in the arithmetical hierarchy.
\begin{table}[t]
	\begin{center}
	\begin{tabular}{c|cccc}
				&	Satisfiability	&	Unsatisfiability	&	Validity	&	Invalidity	\\
		\hline~\\[-1ex]
		$\forall^*$-fragment	&	$\Pi^0_1$-complete	&	$\Sigma^0_1$-complete	&	$\Sigma^0_0$		&	$\Sigma^0_0$	\\[0.5ex]
		$\exists^*$-fragment	&	$\Sigma^0_0$		&	$\Sigma^0_0$		&	$\Sigma^0_1$-complete	&	$\Pi^0_1$-complete	
	\end{tabular}
	\end{center}
	\caption{Overview regarding the degree of unsolvability of the (un)satisfiability and (in)validity problems for the purely universal and purely existential fragment of Presburger arithmetic with additional uninterpreted predicate symbols. Notice that membership in $\Sigma^0_0$ (which coincides with $\Pi^0_0$) entails decidability of the respective problem.}
	\label{table:SatAndVal}
\end{table}

\section{One $\forall\exists$ quantifier alternation leads to $\Sigma^1_1$-completeness}\label{section:SigmaOneOneCompleteness}

\bigskip
Halpern has shown that the satisfiability problem for Presburger arithmetic with any choice of additional uninterpreted function symbols and predicate symbols lies in $\Sigma^1_1$ (Theorem~3.1 in \cite{Halpern1991}). 
This result is independent of the number of occurring quantifier alternations.
In the present section, we show that already a single quantifier alternation suffices to make the problem complete for $\Sigma^1_1$.
We leverage the following result, due to Alur and Henzinger.
\begin{proposition}\label{proposition:SigmaOneOneHardnessOfRecurrence}[Lemma 8 in \cite{Alur1994}]
	The problem of deciding whether a given nondeterministic two-counter machine has a recurring computation is $\Sigma^1_1$-hard.
\end{proposition}

A \emph{nondeterministic two-counter machine} differs from the deterministic model described in Section~\ref{section:Preliminaries} in that it allows nondeterministic branching after a program line has been executed. This means that after the execution of a program line $j$ (which does not result in a jump induced by a \texttt{test\&dec} instruction) the machine does not necessarily proceed to the $(j+1)$-st line, but may have the choice between two specified options. 

This kind of nondeterminism can easily be incorporated into the encoding presented in Section~\ref{section:Encoding}.
For instance, the nondeterministic version of the instruction $j:\texttt{inc}(C_1)$ can be represented by the formula
\begin{align*} 
			\forall x y z.\; &\phiChunk(x) \;\wedge\; 2x \leq y \;\wedge\; y \leq 3x \;\wedge\; \phiEnd(y) \;\wedge\; 3x \leq z \;\wedge\; z \leq 4x \;\wedge\; \phiEnd(z) \;\wedge\; \chi_j(x) \\
			&\;\;\longrightarrow\;\; \phiEnd(6x+y+1) \;\wedge\; \phiEnd(9x+z) \;\wedge\; \bigl( \chi_{j'}(4x) \vee \chi_{j''}(4x) \bigr) ~.
\end{align*}
The last conjunct $\bigl( \chi_{j'}(4x) \vee \chi_{j''}(4x) \bigr)$ now offers a choice between program lines $j'$ and $j''$ as the ones that are to be executed next.

Consequently, we can reuse major parts of our encoding in order to prove $\Sigma^1_1$-hardness.
For any nondeterministic two-counter machine $\cM$ we write $\varphi'_\cM$ to address the encoding of $\cM$'s program in accordance with Section~\ref{section:Encoding} and the just described adaptations due to the nondeterministic setting.

A computation performed by a nondeterministic two-counter machine is considered to be \emph{recurring} if and only if it starts with both counters set to zero and if it reaches the program line with address $0$ infinitely often.
This means, we have to remove $\varphi_4$ from the encoding set of sentences and replace it with a proper formalization of the recurrence condition:
	\[ \varphi'_5 :=\;\; \forall x \exists y.\; x \leq y \;\wedge\; \phiChunk(y) \;\wedge\; \chi_0(y) ~. \]
$\varphi'_5$ formulates recurrence in a positive fashion by saying that at any point in time program line $0$ will be reached eventually. 
Formally speaking, $\varphi'_5$ exhibits a $\forall \exists \exists$ quantifier prefix, since $x \leq y$ is an abbreviation for $\exists  z.\, x + z = y$.

Finally, in order to account for the specific input requirements posed in the definition of recurrence, we construct $\varphi'_2$ from $\varphi_2$ by setting $m = n = 0$, i.e.\
 	\[ \varphi'_2 :=\;\; \phiEnd(2 d+3) \wedge \phiEnd(3 d+3) \wedge \chi_0(d) ~.\]	
 	
\begin{lemma}
	The nondeterministic two-counter machine $\cM$ has a recurring run if and only if $\varphi_1 \wedge \varphi'_2 \wedge \varphi_3 \wedge \varphi'_\cM \wedge \varphi'_5$ is satisfiable.
\end{lemma}
By Proposition~\ref{proposition:SigmaOneOneHardnessOfRecurrence}, this yields $\Sigma^1_1$-hardness. 
Due to the result by Halpern \cite{Halpern1991}, we know that the set of satisfiable Presburger arithmetic sentences with additional uninterpreted predicate symbols lies in $\Sigma^1_1$.
Hence, the following theorem holds.
\begin{theorem}\label{theorem:SigmaOneOneCompleteness}
	The set of satisfiable sentences of the $(\forall^* \wedge \forall\exists^2)$-fragment of Presburger arithmetic with a single additional uninterpreted unary predicate symbol is $\Sigma^1_1$-complete.
\end{theorem}
Due to the strong relations between satisfiability in one fragment and invalidity in its dual, the above theorem entails $\Pi^1_1$-completeness of the set of invalid sentences in the $(\exists^* \vee \exists\forall^2)$-fragment.

Moreover, notice that the theorem can be reformulated in terms of uninterpreted unary functions instead of uninterpreted unary predicates. However, in contrast to Theorem~\ref{theorem:UndecidabilityWithFunctions}, we lose the property that the encoding results in a set of Horn clauses when transformed into CNF. The reason is the involved nondeterminism and the way we have encoded nondeterministic branching.

Over the domain of the reals, we can only show $\Sigma^1_1$-hardness of the satisfiability problem, since Halpern's upper bound does only cover the realm of the natural numbers.

\begin{theorem}
	The set of satisfiable sentences of the $(\forall^* \wedge \forall\exists^2)$-fragment of linear arithmetic over the reals with a single additional uninterpreted unary predicate symbol is $\Sigma^1_1$-hard.
\end{theorem}

\section{On the relevance to verification}\label{section:Verification}

In verification one usually abstracts from some of the limitations that apply to real-world computing devices.
In particular, memory is often regarded as an infinite resource in one way or another. 
This can be due to infinitely many memory locations---similarly to the infinite tape of a Turing machine---or due to the capability of storing arbitrarily large integers in one memory location---similarly to the counters of counter machines.

In our encoding of two-counter machines the uninterpreted predicate symbol $P$ serves as a representation of an unbounded memory.
As we have pointed out, any interpretation $P^\cI \subseteq \Nat$ can be conceived as an infinite sequence of bits. And these bits can be accessed by integer addresses. We have also pointed out in Section~\ref{section:UndecidabilityWithFunctionSymbols} that the same applies to uninterpreted function symbols over the integers or some co-domain with at least two distinct elements.

Clearly, this means that our results are relevant to all verification approaches in which an infinite memory is modeled and in which there are sufficiently strong means available to access individual memory locations. We shall discuss several exemplary settings: separation logic over an integer-indexed heap, logics formalizing integer-indexed arrays or similar data structures, logics with restricted forms of linear integer arithmetic.

Verification is one driving force behind attempts to the combination of theories, such as integer or real arithmetic and the \emph{theory of equality over uninterpreted functions (EUF)}.\footnote{For references see, e.g., Chapter 10 in \cite{Bradley2007b}, or Chapters 10 and 12 in \cite{Kroening2016}.}
For quantifier-free cases the Nelson-Oppen procedure \cite{NelsonOpppen79} provides a general-purpose framework that yields a decision procedure for combined theories from decision procedures for the constituent theories.
Over the course of the last decade numerous approaches have been proposed to go beyond the quantifier-free setting and handle quantification \cite{Flanagan2004, Detlefs2005, Ge2009a, Ge2009b, Bjorner2013, Reynolds2014, Reynolds2015}. 
Typically, some kind of heuristic is applied to guide instantiation to ground formulas. Often the methods are incomplete in the sense that unsatisfiable sentences are not necessarily recognized as such.
Nevertheless, the proposed methods have been implemented and successfully applied, e.g.\ in the tools Verifun, Simplify, and the CVC family.

These approaches inevitably face undecidability when they allow too liberal syntax that combines arithmetic with uninterpreted function or predicate symbols.
The hardness results presented in this paper draw a sharp line around what is possible in such settings.
In the remaining sections, we give reasons why incomplete heuristics is sometimes the best one can expect.

\subsection{Separation logic}\label{section:VerificationSeparationLogic}
In \cite{Reynolds2017} the Bernays--Sch\"onfinkel--Ramsey fragment ($\exists^*\forall^*$-sentences) of separation logic is investigated. 
The quantifiers range over memory locations.
Although Reynolds, Iosif, and Serban also present a refinement of Halpern's undecidability result \cite{Halpern1991} for Presburger arithmetic with an additional uninterpreted unary predicate symbol, their approach differs from ours in an important aspect.
In their setting it is sufficient to consider models wherein the unary predicate is interpreted over \emph{finite} subsets of $\Nat$.
In our case finite subsets do not suffice. 
It is due to this difference, that their strategy can be used to also show undecidability of the satisfiability problem for $\exists^*\forall^*$-sentences of separation logic over a heap with \emph{finitely} many integer-indexed memory locations, each capable of storing one integer of arbitrary size.

Our results in Sections~\ref{section:Encoding} and \ref{section:SigmaOneOneCompleteness} have implications for settings with integer-indexed heaps that comprise a countably infinite number of memory locations, each capable of distinguishing at least two values (e.g.\ $0$ and $1$) or states (e.g.\ \emph{allocated} and \emph{not allocated}).
However, a slight modification of the encoding in Section~\ref{section:FirstEncoding} leads to a result that subsumes Theorem~3 in \cite{Reynolds2017} and also entails undecidability of the satisfiability problem for the $\exists^*\forall^*$-fragment of separation logic with integer-indexed heaps that comprise finitely many memory locations, each capable of storing at least one bit of information.
\begin{lemma}\label{lemma:UndecidabilityOverFiniteSubsets}
	Let $\cM$ be a two-counter machine with $K+1$ program lines, labeled $0, \ldots, K$, and let $\<m,n\>$ be a pair of nonnegative integers.
	There is a sentence $\varphi$ from the $(\exists\forall^*)$-fragment of Presburger arithmetic with an additional unary predicate symbol $P$, such that the following statements are equivalent:
		\begin{enumerate}[label=(\alph{*}), ref=(\alph{*})]
			\item $\varphi$ is satisfied by an interpretation $\cI$ under which $P^\cI$ is a finite subset of $\Nat$,
			\item $\cM$ reaches the \texttt{halt} instruction when started on the input $\<m,n\>$.
		\end{enumerate}	
\end{lemma}
\begin{proof}[Proof sketch]
	Let $\varphi''_\cM$ be the encoding of $\cM$'s program in accordance with Section~\ref{section:FirstEncoding} with the exception that we do not encode the instruction in program line $K$. Due to our conventions, this program line contains the \texttt{halt} instruction.
	
	Let $\varphi''_1(z)$ be the result of replacing the subformula (\ref{formula:III}) in $\varphi_1$ with
		\[ \forall x.\; x < z \;\wedge\; \phiChunk(x) \;\;\longrightarrow\;\; \phiBreak(2x) \;\wedge\; \phiBreak(3x) \;\wedge\; \phiChunk(4x) ~.\]	
	Moreover, let 
		\[ \varphi''_4(z) :=\;\; \phiChunk(z) \wedge \chi_K(z) ~.\]		
	Notice that both formulas $\varphi''_1(z)$ and $\varphi''_4(z)$ contain the free variable $z$.
	We now set 
		\[ \varphi :=\;\; \exists z.\; \varphi''_1(z) \wedge \varphi_2 \wedge \varphi_3 \wedge \varphi''_\cM \wedge \varphi''_4(z) ~. \]
	
	There exists a model $\cI'$ of $\varphi$ if and only if $\cM$ reaches program line $K$ when started on the input $\<m,n\>$.
	Due to the modifications in $\varphi''_1$, the formula $\phiChunk(x)$ does not have to be satisfied for arbitrarily large values of $x$.
	One consequence is that the run of $\cM$ represented by a model of $\varphi$ can be aborted at the point when program line $K$ is reached. This means, in contrast to the proof of Lemma~\ref{lemma:SimulationOfComputations}, we do not have to artificially continue $\cM$'s run beyond that point. Hence, any model of $\varphi$ can be modified in such a way that from a certain point on the bit sequence represented by the interpretation of $P$ contains only zeros.
\end{proof}

\subsection{Verification of data structures}\label{section:VerificationDataStructures}
There are undecidability results in the context of verification of programs that use integer-indexed arrays as data structures.
Examples can be found in \cite{Bradley2006} (Section~5), \cite{Bradley2007a} (Sections~2.4 and 2.6.3), \cite{Habermehl2008} (Section 3). 
The reductions presented therein are based on arrays with infinite co-domains, such as the integers or the reals.
Moreover, they typically use at least one quantifier alternation (but face other restrictions of syntax).
Usually, several arrays are used for convenience, but could be merged into one. 
For our proof approach a single array is sufficient as well.

Read operations on integer-indexed arrays can be formalized as uninterpreted functions with an integer domain.
Hence, our results, Theorems~\ref{theorem:UndecidabilityWithFunctions} and \ref{theorem:UndecidabilityWithFunctionsOverReals} in particular, show that reasoning about integer- or real-indexed arrays over a \emph{finite co-domain} with at least two elements can lead to undecidability, if constraints on array indices allow the necessary syntactic means.
Notice that for the proof it is not necessary to have write operations on arrays. This means, a single integer-indexed read-only array over a Boolean co-domain suffices.

The mentioned results and arguments hold for arrays that comprise an infinite number of elements.
However, due to Lemma~\ref{lemma:UndecidabilityOverFiniteSubsets}, undecidability arises also in the context of finite arrays (over finite co-domains), as long as their length is not bounded by a concrete number.

\begin{remark}
	The above arguments are also applicable to recursively defined data structures, such as lists or trees, as soon as there are sufficiently strong syntactic means available to access the stored information. This means, if one can essentially simulate arrays using a recursive data structure, then our results apply immediately.
	Examples of such setting are lists where the stored elements can be addressed by integers, or where one can access the sublist starting at the position that is $x$ nodes away from the head (for some integer-sort variable $x$ that may be universally quantified). 
\end{remark}	

\subsection{Verification using counter arithmetic}\label{section:VerificationEUFwithCounterArithmetic}

In \cite{Bryant2002} the fragment \emph{CLU} is introduced, which constitutes a strongly restricted fragment of Presburger arithmetic with additional uninterpreted function and predicate symbols. A less syntactically sugared subfragment is treated in \cite{Ganzinger2004} and in \cite{Armando2009}. Regarding arithmetic, there are only two operators available in CLU: the successor operator \texttt{succ} and the predecessor operator \texttt{pred}. The language of CLU does not contain an interpreted constant that addresses zero. On the other hand, some syntactic elements are added for convenience, such as lambda abstraction and an \texttt{if-then-else} operator.
The fragment was chosen for its expressiveness and the fact that it facilitates efficient reasoning.
Although quantifier-free in its original definition, the authors state about their verification tool \emph{UCLID} that they ``have built some support for quantifiers in CLU using automatic quantifier instantiation heuristics'' (\cite{Bryant2002}, Section~7).

In what follows, we consider the extension of CLU by universal quantification of integer variables.
We shall refer to this extended language as \emph{uCLU}.
By a result due to Gurevich \cite{Gurevich1976} (see also \cite{Borger1997}, Theorems 4.1.8 and 4.1.11), satisfiability of EUF sentences with universal quantification is undecidable. Hence, satisfiability of uCLU sentences is undecidable as well.
\begin{proposition}[corollary of the Main Theorem in \cite{Gurevich1976}]
	(Un)satisfiability for uCLU sentences is undecidable.
\end{proposition}
On the other hand, the unsatisfiable sentences of pure first-order logic (and thus also of quantified EUF) are recursively enumerable.
We next argue that uCLU does not possess this property.

The encoding of two-counter machines from Section~\ref{section:Encoding} and \ref{section:SigmaOneOneCompleteness} cannot immediately be translated into uCLU.
First of all, we need to fix a point of reference that serves as zero (CLU does not contain $0$ as a built-in constant).
Moreover, expressions of the form $k x$ for some integer $k$ and some integer-sort variable $x$ require a form of addition that is not available as a built-in operation in uCLU. 
However, with unrestricted universal quantification over integer variables at hand, we can easily define addition as a  function.
Hence, we need only the following uninterpreted symbols to encode two-counter machines: one constant $c_0$ serving as zero, one binary function realizing addition, one uninterpreted unary function or predicate symbol serving as memory.

We define the addition function as follows, where we use $c_0$ as zero:
	\[
		\begin{array}{l@{\hspace{1ex}}rl}
			\forall x.\; &					&\Add(x,c_0) = x \\
			\forall x y.\; &\Succ(y) > c_0 \;\longrightarrow	&\Add\bigl( x,\Succ(y) \bigr) = \Add \bigl( \Succ(x),y \bigr) \\
			\forall x y.\; &\Succ(y) < c_0 \;\longrightarrow	&\Add\bigl( x,\Succ(y) \bigr) = x ~.
		\end{array}	
	\]

All constants that we use in the encoding shall be written as $\Succ^k(c_0)$ instead of just $k$.
Moreover, we add guards $x \geq c_0 \rightarrow \ldots$ to each sentence for every universally quantified variable $x$ that occurs in the sentence. 

As we have seen in Section~\ref{section:SigmaOneOneCompleteness}, in particular in Theorem~\ref{theorem:SigmaOneOneCompleteness}, $\forall\exists$ quantifier alternations lead to (un)satisfiability problems that are not even recursively enumerable. 
Since CLU allows uninterpreted function symbols, uCLU essentially allows $\forall^*\exists^*$ quantifier prefixes. 
Hence, we may introduce a fresh unary Skolem function $f_\text{init}$ and translate the sentence $\varphi'_5$ from Section~\ref{section:SigmaOneOneCompleteness} into the uCLU formula 
	\[ \forall x.\, x \geq 0 \;\longrightarrow\; x \leq f_\text{init}(x) \;\wedge\; \phiChunk\bigl( f_\text{init}(x) \bigr) \;\wedge\; \chi_0\bigl( f_\text{init}(x) \bigr) ~.\]
This means, we can transfer Theorem~\ref{theorem:SigmaOneOneCompleteness} to uCLU and thus obtain the following result.
\begin{proposition}\label{proposition:uCLUhardness}
	Neither the set of satisfiable uCLU sentences nor the set of unsatisfiable uCLU sentences is recursively enumerable.
	In particular, there cannot be any sound and refutationally complete calculus for uCLU.
\end{proposition}

In \cite{Armando2009} the authors present a combination result (Theorem~4.6) for the ground theories of \emph{integer-offsets} (the arithmetic subfragment of CLU embodied by the operators \texttt{succ} and \texttt{pred}), arrays, and/or EUF (as long as the signature of uninterpreted functions does not contain the array sort). The result states that the satisfiability of sentences in such combined theories can be decided using term-rewriting methods.
By a similar line of argument that led us to Proposition~\ref{proposition:uCLUhardness}, it follows that Theorem~4.6 in \cite{Armando2009} cannot be generalized to cases which admit quantification over integer-sort variables. But we do not only lose decidability, we also lose semi-decidability. In other words, it is impossible to devise sound and complete calculi for combinations of EUF and arithmetic---even in such a restricted form as in CLU---if universal quantification of integer variables is available.

\begin{remark}
	Note that lists plus an append, a length operator, and a $<$ predicate can be used to simulate natural numbers with addition.\footnote{Alternatively, the relation $<$ could be defined using equality testing on lists and existential quantification over list---in the same manner as we have defined $<$ in Presburger arithmetic in Section~\ref{section:basicDefinitionsPresburger}.}
	Hence, combining such a theory with EUF leads to undecidability, if universal quantification is allowed.
	
	Similarly, sets with disjoint union (or standard union plus a disjointness predicate or a membership predicate) and a cardinality function can simulate natural numbers with addition.
\end{remark}

\subsection{Almost uninterpreted formulas with offsets}\label{section:VerificationEUFwithOffsets}

In \cite{Ge2009b} Ge and de Moura define the fragment of \emph{almost uninterpreted formulas}.
It constitutes a combination of subfragments of first-order logic, EUF, and linear arithmetic over the integers.
Its language admits uninterpreted predicate symbols, function symbols and constant symbols.
Formulas are assumed to be given in CNF.
All variables are universally quantified, but may only occur as arguments of uninterpreted function or predicate symbols with the following exceptions.
Literals of the form $\neg (x \leq y)$, $\neg (x \leq t)$, $\neg (x \geq t)$, $\neg (x = t)$, $\neg (x \leq y + t)$, $x = t$ with integer-sort variables $x,y$ are allowed for all ground terms $t$ of the integer sort.
Moreover, terms of the form $f(\ldots, x + t, \ldots)$ and $P(\ldots, x + t, \ldots)$ are allowed for ground terms $t$ of the integer sort, function symbols $f$ and predicate symbols $P$.
In what follows we shall be more liberal with the syntax than this.
However, the formulas that we will present can be rewritten into equivalent ones that obey the above restrictions.

The encoding of two-counter machines in Section~\ref{section:FirstEncoding} requires different syntactic means than the ones available in Ge and de Moura's \emph{almost uninterpreted} fragment.
Hence, a proof of undecidability in the syntax of \cite{Ge2009b} needs a slight shift of paradigm.
We start from the reduced form outlined in Section~\ref{section:ReducedEncoding}, since this encoding requires at most two integer-sort variables in atomic constraints.
In our previous encodings the length of the chunks (substrings) storing a single configuration $\<\ell, c_1, c_2\>$ increases over time.
This behavior is necessary to formalize non-terminating runs---and recurring runs in particular---by satisfiable formulas. 
However, in order to formalize a run that eventually reaches the \texttt{halt} instruction by a satisfiable sentence, it suffices to fix the length of the chunks representing a single configuration to a size that can accommodate all configurations that occur in the run, depending on the machine program and on the given input. 
In Ge and de Moura's fragment uninterpreted constant symbols are available that can be used for this purpose.
In what follows, the uninterpreted constant $d$ is used to determine the length of subchunks, as depicted below.
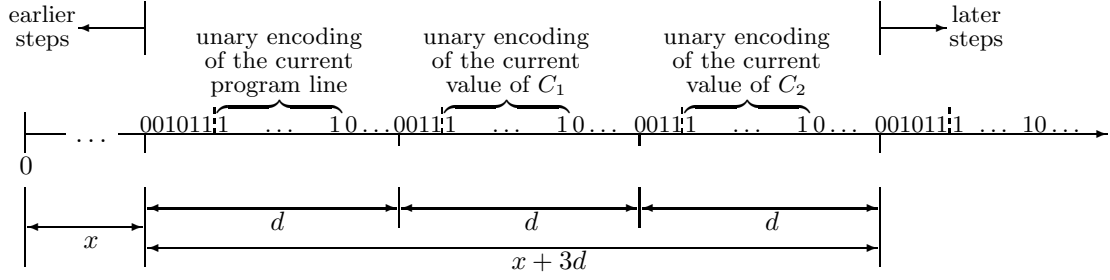
\begin{center}
	\begin{picture}(410,100)
		\put(-2,35){\mbox{$0$}}
		
		\put(0,50){\line(1,0){15}}
		\put(19,47){\mbox{$\ldots$}}
		\put(35,50){\vector(1,0){370}}
		
		\put(0,43){\line(0,1){15}}
		\put(45,43){\line(0,1){7}}
		\put(140,46){\line(0,1){4}}
		\put(230,46){\line(0,1){4}}
		\put(320,43){\line(0,1){7}}
		
		\put(43,51){\mbox{\small$001011$}}
		\put(138,51){\mbox{\small$0011$}}
		\put(228,51){\mbox{\small$0011$}}
		\put(318,51){\mbox{\small$001011$}}
		
		\multiput(71,50)(0,4){3}{\line(0,1){2}}
		\multiput(156,50)(0,4){3}{\line(0,1){2}}
		\multiput(246,50)(0,4){3}{\line(0,1){2}}
		\multiput(346,50)(0,4){3}{\line(0,1){2}}
		
		\put(72,51){\mbox{\small$\overbrace{1\hspace{3ex}\ldots\hspace{3ex} 1}0\ldots$}}
		\put(157,51){\mbox{\small$\overbrace{1\hspace{3ex}\ldots\hspace{3ex} 1}0\ldots$}}
		\put(247,51){\mbox{\small$\overbrace{1\hspace{3ex}\ldots\hspace{3ex} 1}0\ldots$}}
		\put(347,51){\mbox{\small$1\hspace{1ex}\ldots\hspace{1ex} 10\ldots$}}
		
		\put(58,75){\small\begin{tabular}{c}unary encoding \\[-0.5ex] of the current\\[-0.5ex] program line \end{tabular}}
		\put(143,75){\small\begin{tabular}{c}unary encoding \\[-0.5ex] of the current\\[-0.5ex] value of $C_1$ \end{tabular}}
		\put(233,75){\small\begin{tabular}{c}unary encoding \\[-0.5ex] of the current\\[-0.5ex] value of $C_2$ \end{tabular}}
		
		\put(45, 80){\line(0,1){20}}
		\put(45,90){\vector(-1,0){25}}
		\put(-12,88){\small\begin{tabular}{c} earlier\\[-0.5ex] steps \end{tabular}}
		\put(320, 80){\line(0,1){20}}
		\put(320,90){\vector(1,0){25}}
		\put(340,88){\small\begin{tabular}{c} later\\[-0.5ex] steps \end{tabular}}
		
		\put(0,0){\line(0,1){30}}
		\put(45,0){\line(0,1){30}}
		\put(140,15){\line(0,1){15}}
		\put(230,15){\line(0,1){15}}
		\put(320,0){\line(0,1){30}}
		
		\put(1,15){\vector(1,0){43}}
		\put(44,15){\vector(-1,0){43}}
		\put(22,7){\mbox{$x$}}
		\put(46,7){\vector(1,0){273}}
		\put(319,7){\vector(-1,0){273}}
		\put(182,-1){\mbox{$x+3d$}}
		\put(46,22){\vector(1,0){93}}
		\put(139,22){\vector(-1,0){93}}
		\put(92,14){\mbox{$d$}}
		\put(141,22){\vector(1,0){88}}
		\put(229,22){\vector(-1,0){88}}
		\put(187,14){\mbox{$d$}}
		\put(231,22){\vector(1,0){88}}
		\put(319,22){\vector(-1,0){88}}
		\put(277,14){\mbox{$d$}}
	\end{picture}
\end{center}
Moreover, we now start the encoding of the run at the very first bit of the bit string represented by $P$.
We replace the formula $\varphi_1$ by the following, somewhat simpler formula $\varphi'''_1$.
Let $k$ be the result of the expression $\max(K+6, m+4, n+4)$, where $K$ is the address of the last program line and $m$ and $n$ are the input values.
\begin{align*}
	\varphi'''_1 :=\;\; 
	&d \geq k \;\wedge\; e \geq 0 \\
	&\wedge\; \phiChunk(0) \;\wedge\; \phiBreak(d) \;\wedge\; \phiBreak(2d) \;\wedge\; \bigl( \forall x.\; x < 0 \;\longrightarrow\; \neg P(x) \bigr) \\
	&\wedge\; \bigl( \forall x.\; \phiChunk(x) \;\wedge\; x < 3d \;\wedge\; x \neq 0 \;\longrightarrow\; \bot \bigr) \\
	&\wedge\; \bigl( \forall x.\; \phiBreak(x) \;\wedge\; x < 3d \;\wedge\; x \neq d \;\wedge\; x \neq 2d \;\longrightarrow\; \bot \bigr) \\
	&\wedge\; \bigl( \forall x.\; \phiChunk(x) \;\wedge\; x < e \;\;\longrightarrow\;\; \phiChunk(x+3d) \bigr)\\
	&\wedge\; \bigl( \forall x.\; \phiChunk(x+3d) \;\wedge\; x \geq 0 \;\;\longrightarrow\;\; \phiChunk(x) \bigr) \\
	&\wedge\; \bigl( \forall x.\; \phiBreak(x) \;\wedge\; x < e \;\;\longrightarrow\;\; \phiBreak(x+3d) \bigr) \\
	&\wedge\; \bigl( \forall x.\; \phiBreak(x+3d) \;\wedge\; x \geq 0 \;\;\longrightarrow\;\; \phiBreak(x) \bigr)
\end{align*}
The purpose of the uninterpreted constant $e$ is to mark the end of the run, as we will see later.

The sentences $\varphi_2$ and $\varphi_3$ can be adapted in the same spirit:
	\begin{align*}
		\varphi'''_2 &:=\;\; \chi_0(0) \;\wedge\; \phiEnd(d+3+m) \;\wedge\; \phiEnd(2d+3+n) \\
		\varphi'''_3 &:=\;\;  \forall x y.\; \phiChunk(x) \;\wedge\; \phiEnd(y) \;\wedge\; x + 5 + K < y \;\wedge\; y \leq x+d \;\;\longrightarrow\;\; \bot ~.
	\end{align*}	
The adapted encoding of an instruction \texttt{inc($C_1$)} comprises the formulas
	\begin{align*} 
		\forall x y.\; &\phiChunk(x) \;\wedge\; x+d \leq y \;\wedge\; y \leq x+2d \;\wedge\; \phiEnd(y) \;\wedge\; \chi_j(x) \\
			&\hspace{50ex} \longrightarrow\;\; \phiEnd(y+3d+1) \;\wedge\; \chi_{j+1}(x+3d)\\[1ex]
		\forall x z.\; &\phiChunk(x) \;\wedge\; x+2d \leq z \;\wedge\; z \leq x+3d \;\wedge\; \phiEnd(z) \;\wedge\; \chi_j(x) \;\;\longrightarrow\;\; \phiEnd(z+3d)
	\end{align*}				
The other instructions can be adapted analogously. The only exception is the \texttt{halt} instruction in the last program line which we shall not encode, as in the proof sketch  for Lemma~\ref{lemma:UndecidabilityOverFiniteSubsets}.

Finally, we also have to modify the condition that the two-counter machine halts at some point in time.
We use another uninterpreted constant $e$ for this purpose:
	\[ \varphi_4''' :=\;\; \phiChunk(e) \;\wedge\; \chi_K(e) ~.\]
Consequently, using the fragment given in \cite{Ge2009b}, we can encode the halting problem of a two-counter machine $\cM$ on input $\<m,n\>$ using only a single uninterpreted unary predicate symbol $P$ (or a single function symbol) plus two uninterpreted constant symbols $d,e$.
More precisely, if $\cM$ halts on $\<m,n\>$, then there is model $\cI$ of the encoding sentence such that $P^\cI$ is a finite set of integers.

The outlined formalization is sufficient for a halting run of a two-counter machine. However, we cannot formalize recurring counter machines in this way. Thus, we do not obtain hardness beyond recursive enumerability. Indeed, this is in line with \cite{Ge2009b}, where a refutationally complete calculus is given for the described fragment.

The realm of recursive enumerability can be left easily. For instance, it is sufficient to allow scalar multiplication combined with addition for integer-sort variables, i.e.\ expressions of the form $k x + y$ for integers $k$.
Similarly, it would suffice to admit expressions $g(x) + y$, as we can define, e.g., 
	\[ \Times_k(0) = 0 \;\;\wedge\;\; \forall x.\; x \geq 0 \;\rightarrow\; \Times_k(x+1) = \Times_k(x) + k \]
for any positive integer $k$.
With a syntax extended this way, one could realize the encoding from Section~\ref{section:FirstEncoding}.

\section{Conclusion}

In this paper we have sharpened the known undecidability results for the language of Presburger arithmetic augmented with uninterpreted predicate or function symbols.
We have shown that already the purely universal fragment of such extended languages yields an undecidable satisfiability problem. 
More precisely, we have shown $\Sigma^0_1$-completeness of the corresponding set of unsatisfiable sentences.
In the case of extensions by uninterpreted function symbols, the fragment can even be restricted to Horn clauses while retaining an undecidable satisfiability problem.
Moreover, we have strengthened Halpern's $\Sigma^1_1$-hardness result (Theorem 3.1 in \cite{Halpern1991}) in that we have shown that a single $\forall\exists$ quantifier alternation suffices for a proof. More precisely, the satisfiability problem for $\forall^*\exists^2$-sentences of this extended language is $\Sigma^1_1$-hard. 

In addition, we have transfered the mentioned undecidability and hardness results to the realm of linear arithmetic over the ordered real numbers augmented with a single uninterpreted predicate symbol.

Concerning automated reasoning, we have mentioned in Section~\ref{section:UnsolvabilityDegree} that there are refutationally complete deductive calculi for the purely universal and purely existential fragments of Presburger arithmetic with uninterpreted predicate symbols. In the existential case even decision procedures exist.
On the other hand, the hardness result presented in Section~\ref{section:SigmaOneOneCompleteness} entails that there cannot be sound and refutationally complete calculi that can handle $\forall\exists$ quantifier alternations, if the problem is not restricted any further. The same applies to fragments allowing unrestricted combinations of universal quantification and function symbols, as Skolem functions are at least as powerful as existential quantifiers in this context.

Apart from their theoretical value, the results presented in this paper are relevant for several areas of verification.
In Section~\ref{section:Verification} we have elaborated on the implications for the Bernays--Sch\"onfinkel fragment ($\exists^*\forall^*$-sentences) of separation logic, quantified theories of data structures, arrays in particular, and quantified combinations of the theory of equality over uninterpreted functions with strongly restricted fragments of Presburger arithmetic. Moreover, we have argued that in certain settings we cannot even hope for refutationally complete deductive calculi. In such cases we either have to content ourselves with heuristics instead of sound and complete methods or formulate restricted fragments that lead to less hard (un)satisfiability problems.


\subsection*{Acknowledgments}
	The authors thank Radu Iosif and Stanislav Speranski for inspiring and enlightening discussions.



\end{document}